\documentclass[journal]{IEEEtran}
\IEEEoverridecommandlockouts
\usepackage{setspace}
\usepackage{cite}
\usepackage{amsmath,amssymb,amsfonts,mathtools,bm}
\usepackage{amsthm}
\usepackage{algorithm}
\usepackage{algpseudocode}
\usepackage{graphicx}
\usepackage{textcomp}
\usepackage{xcolor,float}
\usepackage{tabularx}
\usepackage{subcaption}
\usepackage{diagbox}
\usepackage{svg}
\usepackage{booktabs}
\usepackage{makecell}
\usepackage{caption}
\usepackage{sidecap}
\usepackage{microtype}
\usepackage{multirow}
\usepackage{adjustbox}
\usepackage{colortbl}
\usepackage{url}
\usepackage{balance}
\usepackage{bbding}
\usepackage{array}
\usepackage{tikz}
\usepackage[final]{changes}
\usepackage{ifthen}
\usepackage{hyperref}

\newtheorem{theorem}{Theorem}

\newtheorem{proposition}{Proposition}

\graphicspath{{figures/}}

\definecolor{bestred}{RGB}{190,0,0}
\definecolor{secondblue}{RGB}{0,70,180}
\definecolor{improvecyan}{RGB}{0,120,140}
\newboolean{showchanges}
\setboolean{showchanges}{false}
\newcommand{\add}[1]{%
    \ifthenelse{\boolean{showchanges}}%
        {\textcolor{blue}{#1}}%
        {#1\relax}%
}

\begin{document}

\title{PALM: Single-Station Super-Resolved Small-Scale Radio-Map Localization by Path-Atom Matching}

\author{Xiucheng Wang and Nan Cheng%
\thanks{The authors are with the State Key Laboratory of ISN and School of Telecommunications Engineering, Xidian University, Xi'an 710071, China (e-mail: xcwang\_1@stu.xidian.edu.cn; dr.nan.cheng@ieee.org).\par This work was supported by the National Key Research and Development Program of China (2020YFB1807700), and the National Natural Science Foundation of China (NSFC) under Grant No.~62201432.}}

\maketitle

\begin{abstract}
Localizing from a single base station is a longstanding goal, since it removes
the synchronized anchors that geometric methods require. A radio map (RM) answers
a position query from this one-station survey, yet classical RMs store coarse
received power and match it by correlation, ignoring the small-scale path
structure a ray tracer provides. We instead build a small-scale RM and show that
cell identification, rather than candidate generation, is its information-limited
bottleneck. We propose path-atom localization by matching (PALM), which
super-resolves a coarse angle-delay observation into scored atoms and matches
them to a ray-traced RM by an exact marginal likelihood. The score marginalizes
atom reality inside the logarithm, and we prove that the common posterior-scaled
surrogate is a Jensen lower bound whose deficit grows with the number of strong
paths. We match on the absolute delay axis under a clock nuisance, since relative
delays jump across shadowing boundaries, and we prove a unit-gradient law, a
capped miss cost, a minimum-mean-square local centroid, and finite-sample
conformal coverage. On the real DeepMIMO campus scenario, PALM localizes to a
1.7 meter median from a single base station, cuts the ninetieth-percentile error
of received-power RM matching by 34 to 62 percent, and halves the single-snapshot
median to 7 meters.
\end{abstract}


\section{Introduction}
\label{sec:introduction}

Position information is a first-class service in sixth-generation networks,
where it supports beam management, integrated sensing, and the channel knowledge
maps that environment-aware physical layers query at deployment
time~\cite{wang2026rmtutorial,heath2016overview}. Geometric methods recover a position
from time or angle measurements at several synchronized anchors. Many single-cell
deployments lack that infrastructure. Parametric estimators can still localize
from one base station by resolving the multipath angles and
delays~\cite{shahmansoori2018,wymeersch2017}. A radio map (RM) instead serves a position
from a single base station by matching a measurement against a one-station
survey. The map is built once for a fixed base station and a fixed scene, after
which every query reuses it. Ray tracing now
produces such maps at scale, and it returns the channel as a small set of paths
rather than as a single aggregate power value. The promise of this regime is an
RM that carries the small-scale path-level geometry of the channel rather
than one power value per cell. The open question is how to match a coarse field
measurement against it.

\begin{figure*}[!t]
\centering
\includegraphics[width=\linewidth]{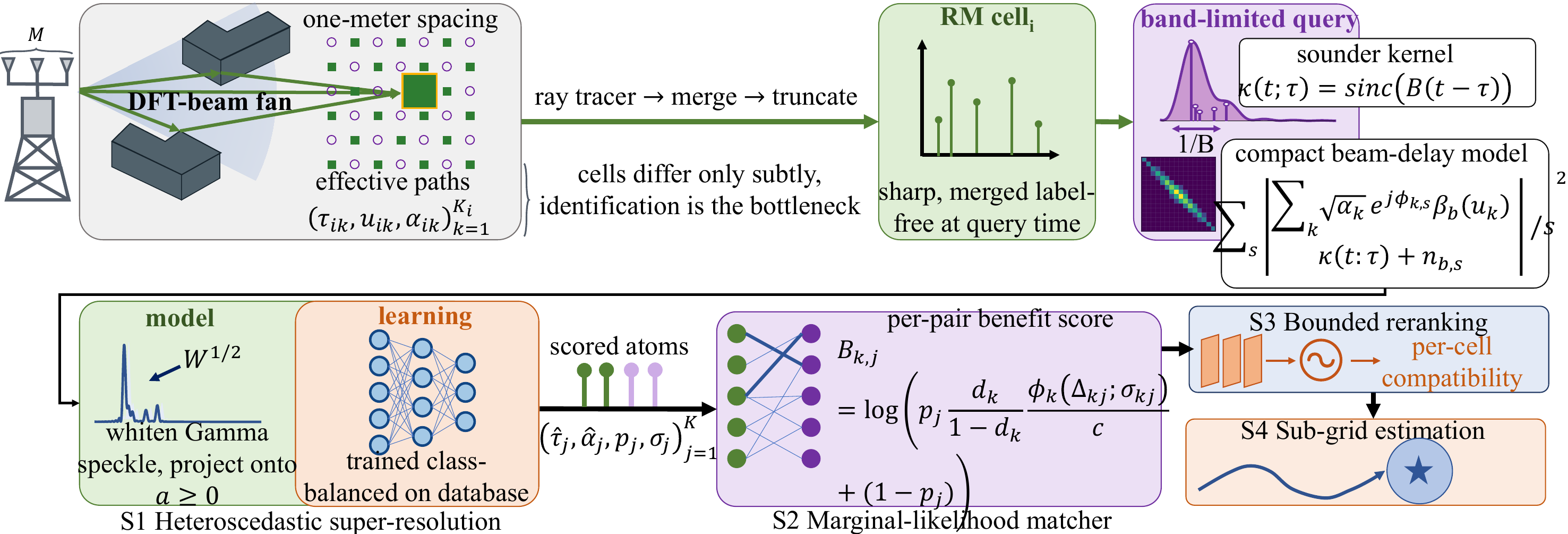}
\caption{Overview of PALM for single-station radio-map localization. (a)~System
model. A single base station surveys the scene once by ray tracing and stores,
for each grid cell, the small-scale path delays, departure angles, and powers of
a radio map. A query observes the same paths through a band-limited aperture, so
the sharp per-cell signature blurs into a noisy power-delay profile and
beam-by-delay map, and cell identification rather than candidate generation
becomes the bottleneck. (b)~PALM pipeline. A heteroscedastic super-resolution
front end and a per-atom verifier recover scored atoms, a marginal-likelihood
matcher scores every database cell on the absolute delay axis under a clock
nuisance, a bounded dense expert reranks the survivors, and a sub-grid centroid
with a Gauss-Newton step returns a position and a split-conformal radius. Each
component is tagged with the theoretical result that justifies it.}
\label{fig:system}
\end{figure*}

A small-scale RM stores, for each surveyed cell, the delays, the
departure angles, and the powers of the effective propagation paths. A query
observes the same paths only through a band-limited aperture. Its power-delay
profile and its angular power spectrum are therefore blurred and noisy
superpositions of the underlying paths. Matching the blurred query to the sharp
map is therefore an inverse problem followed by a retrieval problem. The receiver
must first super-resolve the coarse observation into path atoms, and must then
decide which map cell those atoms came from. The survey model also forbids a
position label at query time, so the matcher must rely on the recovered atoms
alone. A method for this regime must super-resolve the small-scale structure and
must identify the cell robustly under noise.

The central difficulty is cell identification, not candidate generation. A
fine-grid sparse solver generates path atoms densely, so the true paths almost
always lie among the candidates, yet neighbouring cells in a dense survey carry
small-scale signatures that differ only subtly. The signatures are sparse and
similar, a compact early-delay cluster inside a narrow angular band, so the
observable difference between adjacent cells is small. An oracle that retrieves
the true cell among its near neighbours reaches a far higher accuracy than the
deployed matcher, which locates the bottleneck at identification. The ceiling is
set by the measurement rather than by the model, since a richer observation
resolves the ambiguity that a band-limited one cannot. The bottleneck is
therefore information limited, and added model capacity cannot break it.

Existing methods address the wrong half of this problem. Classical RMs
store a coarse received-power value per cell and match it by nearest-neighbour
correlation~\cite{bahl2000radar,he2015wireless}, so they never super-resolve and
they inherit the resolution of the sounder. Subspace and sparse super-resolution
recover path parameters accurately, yet they target the parameters themselves
rather than a map match and assume separations that dense channels
violate~\cite{schmidt1986music,fleury1999sage,candes2014superres}. End-to-end
learned localizers regress a position from a measurement, but a single-scene
survey starves them of data. A higher-capacity network does not improve a match
that is limited by information. None of these methods super-resolves the query
and then scores a small-scale RM match under a principled likelihood. A
matcher that combines super-resolution with a calibrated marginal score has not
been studied for small-scale RM localization.

We propose path-atom localization by matching, which is built on the principle
that identification is the bottleneck and must be scored correctly.
Fig.~\ref{fig:system} overviews the single-station setup and the four-stage
pipeline. PALM
super-resolves the coarse observation into atoms through a heteroscedastic
sparse front end and a learned per-atom verifier. It then scores every map cell
by an exact marginal likelihood that integrates out atom reality inside the
logarithm. It matches on the absolute delay axis under a clock-bias nuisance,
because relative delays jump where the first arrival changes across a shadowing
boundary. A learned dense expert enters the score as a bounded reranking bonus,
a temperature-calibrated local centroid and a Gauss-Newton step place the
estimate below the grid, and a split-conformal radius reports calibrated
uncertainty. The marginal score and the absolute-axis match are the conceptual
core, and each component is justified by a result that fails without it.

The contributions of this paper are theoretical as well as empirical. We
construct the cell score as an exact marginal likelihood and prove that the
common posterior-scaled surrogate is a Jensen lower bound whose deficit grows
with the number of strong paths. We characterize the identifiability of the
delay observables, which justifies matching on the absolute axis, and we bound
the influence of a missed path, the optimality of the local centroid, and the
coverage of the conformal radius. Each theoretical claim is matched to one
component of the pipeline and to a measured confirmation in the experiments. The
main contributions are summarized as follows.
\begin{enumerate}
  \item We localize against a small-scale RM that stores the per-cell delay,
        angle, and power structure, rather than the received-signal-strength
        signature that prior radio-map localization matches. We show that this
        small-scale structure is the substrate a localizer should match, and we
        identify cell identification, rather than candidate generation, as its
        information-limited bottleneck.
  \item We propose PALM, which super-resolves a coarse angle-delay observation
        into scored atoms, matches them by an exact marginal likelihood on the
        absolute delay axis under a clock nuisance, and fuses a bounded dense
        expert before sub-grid refinement.
  \item We prove that the marginal score is exact and that its posterior-scaled
        surrogate is a rank-inverting Jensen bound, and we establish a
        unit-gradient delay law, a bounded-influence capped miss cost, a
        restricted minimum-mean-square centroid, and finite-sample conformal
        coverage.
  \item On the real ray-traced DeepMIMO campus scenario, PALM localizes from a
        single base station to a 1.7 meter median, cuts the ninetieth-percentile
        error of received-power RM matching by 34 to 62 percent across
        signal-to-noise ratios, and halves the single-snapshot median from 14.0
        to 7.0 meters.
\end{enumerate}

\section{Related Work}
\label{sec:related}

Radio-map localization matches a measurement against a survey map of
location-tagged channel signatures. Received-signal-strength maps store a coarse
power vector per cell and localize by nearest-neighbour or weighted
nearest-neighbour search~\cite{bahl2000radar,he2015wireless}, the classical
fingerprinting approach. Probabilistic maps instead replace the search with a
likelihood over a Gaussian signature model~\cite{youssef2005horus}. These methods
are simple and robust, yet their signature is a blurred aggregate that discards
the small-scale path structure of the channel and inherits the resolution of the
sounder. They also snap a query to the nearest surveyed grid point, so their
accuracy is bounded by the grid spacing at high signal-to-noise ratio. PALM keeps
the survey-then-match structure but replaces the aggregate signature with
super-resolved path atoms and a principled likelihood, which recovers structure
the aggregate hides.

Parametric super-resolution estimates path delays and angles from the channel.
Subspace estimators resolve well-separated paths from the covariance
eigenstructure, and expectation-maximization estimators extend them to joint
delay and angle~\cite{schmidt1986music,roy1989esprit,fleury1999sage}. Sparse and
off-grid methods cast recovery as support estimation over a delay dictionary and
recover the support under a minimum-separation
condition~\cite{tropp2007omp,bajwa2010ccs,candes2014superres,tang2013offgrid}. These methods
target the path parameters as the end product and assume separations that dense
channels routinely violate. PALM uses a sparse front end only to propose atoms,
and it shifts the objective from parameter recovery to RM matching, where
the residual ambiguity between near cells is decided by a learned score.

Learning-based localization regresses a position from a measurement or learns a
representation of the radio environment. Deep localization networks map a
channel image to coordinates, and contrastive encoders learn environment
features without labels~\cite{oord2018cpc,chen2020simclr}. Channel charting
embeds channels into a relative map of the radio environment without position
labels~\cite{studer2018charting}. The channel knowledge
map literature constructs a dense radio map across a scene from sparse samples
or scene geometry~\cite{wang2026rmtutorial,levie2021radiounet}. Recent
generative diffusion models synthesize physically consistent maps, including
multipath-aware and inverse
variants~\cite{wang2025radiodiff,wang2026radiodiffk2,wang2026radiodiffinv}. These
methods construct the map as the end product, whereas PALM consumes a ray-traced
map and matches its small-scale structure to localize. They also either demand
more data than a single survey offers or model aggregates rather than the
per-atom match a position query needs. PALM uses
learning where it helps a decision under fixed information, namely a per-atom
reality verifier and a dense reranking expert, and it keeps the explicit
generative match that gives the score its inductive bias. The experiments show
that a higher-capacity learned scorer does not beat this lightweight pairing,
which is consistent with an information-limited bottleneck.

A separate line studies the verification of generated candidates and the
calibration of learned scores. Detection and retrieval systems propose
candidates and rerank them~\cite{carion2020detr}, and a verifier trained against
the task metric turns selection into a supervised problem~\cite{ye2012fmeasure}. Calibration matters because a selection
acts on a posterior, and a miscalibrated posterior shifts the operating
point~\cite{guo2017calibration,platt1999probabilistic}. Conformal prediction
turns a point estimate into a set with finite-sample
coverage~\cite{vovk2005algorithmic}. PALM imports the generate-then-verify
pattern into the front end and reports a split-conformal radius, and it makes the
match itself a calibrated likelihood rather than a tuned threshold. The
combination of marginal-likelihood matching, absolute-axis identifiability, and
conformal uncertainty has not been studied for small-scale RM
localization.

\section{System Model and Problem Formulation}
\label{sec:model}

We consider a single base station that sounds the channel to a grid of receiver
positions in a fixed scene. For receiver $i$, the ray tracer returns a set of
paths, each with a delay, a departure azimuth, and a power. The effective path
set is obtained by a dynamic-range cut, a first-arrival alignment, power-weighted
merging within the resolution cell, and truncation to the strongest components.
The merged set is the RM entry of cell $i$, denoted
$\{(\tau_{ik}, u_{ik}, a_{ik})\}_{k=1}^{K_i}$. Here $\tau_{ik}$ is a relative
delay in nanoseconds, $u_{ik}\in[-1,1]$ is the sine of the departure azimuth,
and $a_{ik}$ is a linear power normalized to the strongest path. The database
also stores the position $\bm{x}_i\in\mathbb{R}^2$ and the absolute first-arrival
time $t_{0,i}$ of each cell. The number of effective paths $K_i$ is small and
varies across cells, and the survey is taken once for the fixed scene.

The sounder observes the paths through a band-limited aperture. The base station
transmits over a bandwidth $B$ and forms a power-delay profile on a uniform delay
grid. Each path contributes a sinc kernel centred at its delay,
\begin{align}
  \kappa(t;\tau) = \operatorname{sinc}\!\big(B\,(t-\tau)\big),
  \label{eq:kernel}
\end{align}
whose first null sits at $1/B$, the Rayleigh resolution of the sounder. Over $S$
snapshots, each path draws an independent uniform phase, the receiver sums the
complex amplitudes, adds circular complex Gaussian noise, and averages the
squared magnitude. The observed profile is
\begin{align}
  y(t) = \frac{1}{S}\sum_{s=1}^{S}
         \Big| \sum_{k} \sqrt{a_k}\, e^{j\phi_{k,s}}\,\kappa(t;\tau_k)
         + n_s(t) \Big|^2 ,
  \label{eq:obs}
\end{align}
with $\phi_{k,s}$ the per-snapshot phase and $n_s$ the noise. The forward model
radiates the raw unmerged rays rather than the merged labels, so the observation
does not commit the inverse crime against the RM. The angle-delay
observation places a half-wavelength uniform linear array (ULA) of $M$ elements
at the base station and beamforms each snapshot with a unitary discrete Fourier
transform (DFT). A path with beamspace coordinate $u_k$ illuminates beam $b$
through the array response $\beta_b(u_k)$, and the averaged beam-by-delay power
map is
\begin{align}
  Y(b,t) = \frac{1}{S}\sum_{s=1}^{S}
           \Big| \sum_{k} \sqrt{a_k}\, e^{j\phi_{k,s}}\,
           \beta_b(u_k)\,\kappa(t;\tau_k) + n_{b,s}(t) \Big|^2 .
  \label{eq:obs2d}
\end{align}
The signal-to-noise ratio (SNR) is defined against the clean incoherent profile
peak. The marginals of $Y$ are an angular power spectrum and a power-delay
profile, and their joint map retains the coupling that the marginals discard.

The front end converts the coarse observation into discrete path atoms. It emits
delay atoms $\{(\hat\tau_j,\hat a_j,p_j,\sigma_j)\}_{j=1}^{J}$ and angle atoms of
the same form, where $p_j$ is a learned reality posterior that atom $j$ matches a
true path and $\sigma_j$ is a calibrated location scale. We write $d_k\in(0,1)$
for the detection probability assigned to path $k$ from its amplitude, $c$ for
the clutter intensity of spurious atoms, and $\varphi_\nu(\cdot;\sigma)$ for the
Student-t density with $\nu$ degrees of freedom and scale $\sigma$. The matcher
operates in one of two synchronization modes. In the unsynchronized mode the
receiver observes only relative delays referenced to its first arrival. In the
synchronized mode the receiver reports a trigger time
$t_{\mathrm{m}} = t_0^\star + b$ with clock bias $b$ of known prior scale
$\sigma_{\mathrm{c}}$, which exposes the absolute delay axis $t_{0,i}+\tau_{ik}$.

The localization task is to estimate the receiver position from the observation
alone. The estimator returns a point $\hat{\bm{x}}$ and a radius $r$ such that
the position error $\lVert\hat{\bm{x}}-\bm{x}^\star\rVert$ is small and the ball
of radius $r$ covers the truth with a prescribed probability. The survey
provides an RM of cells whose small-scale signatures and positions are known, and the
deployment provides no position label. The problem is to identify the cell that
generated the observation and to place the estimate below the survey grid,
using the recovered atoms and, in the synchronized mode, a coarse trigger time.
Table~\ref{tab:notation} summarizes the notation.

\begin{table}[!t]
\centering
\caption{Summary of the main notation.}
\label{tab:notation}
\resizebox{0.9\linewidth}{!}{
\begin{tabular}{@{}c|l@{}}
\hline
Symbol & Meaning \\
\hline
$\bm{x}_i,\ t_{0,i}$ & position and absolute first arrival of cell $i$ \\
$\tau_{ik}, u_{ik}, a_{ik}$ & delay, beamspace angle, power of database path $k$ \\
$y(t),\, Y(b,t)$ & delay profile and beam-by-delay power map \\
$\kappa(t;\tau)$ & band-limited sinc kernel of the sounder \\
$B,\, S,\, M$ & bandwidth, snapshots, array elements \\
$\hat\tau_j, p_j, \sigma_j$ & atom $j$, its reality posterior and location scale \\
$d_k,\ c,\ \varphi_\nu$ & detection probability, clutter intensity, Student-t \\
$b,\ \sigma_{\mathrm c}$ & clock bias and its prior scale (synchronized mode) \\
$\kappa_{\mathrm m},\ T$ & miss-cost cap and centroid temperature \\
\hline
\end{tabular}}
\end{table}

\section{The PALM Framework}
\label{sec:method}

\subsection{Overview}

PALM separates localization into a super-resolution front end, a
marginal-likelihood matcher, and a sub-grid estimator. The front end recovers
scored atoms from the coarse observation, so the true path structure is exposed
to the matcher. The matcher scores every database cell by an exact marginal
likelihood and selects the cell that best explains the atoms. The estimator
converts the discrete decision into a continuous coordinate and a calibrated
radius. The design follows one principle. Super-resolution and candidate
generation are convex and trustworthy and stay model-based, while identification
under fixed information is a discrimination problem and admits a learned aid.
Fig.~\ref{fig:system} overviews the single-station setup and the four pipeline
stages, Algorithm~\ref{alg:palm} states the full pipeline for a single query, and
Fig.~\ref{fig:heatmap} previews how each information level and each estimator
contribute across operating conditions.

\begin{algorithm}[!t]
\caption{PALM localization for one query at deployment.}
\label{alg:palm}
\begin{algorithmic}[1]
\Require Observation $(y,Y)$, database $\{(\bm{x}_i,t_{0,i},\{\tau_{ik},u_{ik},a_{ik}\})\}$, verifiers, dense expert, mode
\Ensure Position estimate $\hat{\bm{x}}$ and conformal radius $r$
\State Super-resolve $(y,Y)$ into atoms $\{(\hat\tau_j,\hat a_j,p_j,\sigma_j)\}$ by heteroscedastic NNLS
\State Score atom reality $p_j$ with the learned verifier
\For{each viable cell $i$}
    \State Solve the assignment $\pi_i$ maximizing the marginal benefit \eqref{eq:score}
    \State In synchronized mode, scan clock hypotheses and gate on $t_{0,i}$
    \State $s_i \gets$ capped marginal score of cell $i$
\EndFor
\State Add the bounded dense bonus to the top-$K$ physics cells: $s_i \mathrel{+}= w\,z_i$
\State Form the local posterior centroid $\bm{x}_{\mathrm c}$ over the ball around $\arg\max_i s_i$
\State $\hat{\bm{x}} \gets$ one Gauss-Newton step on the absolute-axis track gradients from $\bm{x}_{\mathrm c}$
\State $r \gets$ split-conformal radius at level $1-\alpha$
\State \Return $\hat{\bm{x}}, r$
\end{algorithmic}
\end{algorithm}

\subsection{Heteroscedastic Super-Resolution Front End}

The front end recovers delay atoms from the profile and angle atoms from the
spectrum by non-negative least squares (NNLS) against an oversampled dictionary
of squared kernels. Super-resolving a band-limited profile into atoms is a
deconvolution, classically addressed by iterative restoration~\cite{richardson1972}. Averaged power carries multiplicative speckle, since $S$
independent looks make the power at a sample with mean $\mu$ Gamma distributed
with variance $\mu^2/S$. The solve therefore whitens each sample by its modeled
variance and projects onto the non-negative orthant,
\begin{align}
  \hat{\bm{a}} = \arg\min_{\bm{a}\ge 0}
    \big\lVert \bm{W}^{1/2}\big(\bm{\Phi}\bm{a} - y\big)\big\rVert_2^2 ,
  \label{eq:nnls}
\end{align}
where $\bm{\Phi}$ collects the squared-kernel responses with a noise-floor column
and $\bm{W}=\operatorname{diag}(v_i^{-1})$ holds the inverse modeled variances.
Proposition~\ref{prop:gls} shows that the whitened solve is the minimum-variance
unbiased linear estimate on the active support, which prevents strong lobes from
dominating the residual and protects weak true paths. The solver carries a
robust fallback that raises the iteration cap and projects a least-squares
solution when the dictionary is ill conditioned, so a noisy low-SNR profile never
stalls the pipeline. Each recovered atom carries a noise-propagated scale
$\sigma_j$ that the matcher consumes downstream.

A recovered atom may be a true path or a noise artifact, so a learned verifier
attaches a reality posterior to each one. The verifier is a small multilayer
perceptron (MLP) that reads per-atom evidence such as relative power, isolation,
peak shape, and the recovered scale, conditioned on the operating SNR and
snapshot count. It is trained on database cells only, so a query position never
leaks into training. The verifier is trained class balanced, because true atoms
are a small minority of the candidates. An unbalanced classifier otherwise
collapses its posterior onto the prior and erases the true-versus-junk separation
that the match needs. The verifier output $p_j$ enters the matching score as the reality
probability of atom $j$. The front end is thus a convex recovery followed by a
calibrated discrimination, and neither part can hallucinate a path that a
generative model might.

\subsection{Marginal-Likelihood Matching}

Given the atoms and a candidate cell, the matcher scores the cell under a
generative model with four independent mechanisms. Each path is detected
independently with probability $d_k$, each detected path emits one atom with
Student-t location error, clutter atoms arrive as a Poisson process with
intensity $c$, and the front end labels each atom with a reality posterior $p_j$
that is conditionally independent across atoms. Marginalizing the atom reality
inside the logarithm gives the cell score of Proposition~\ref{prop:marginal}.
The benefit of a matched pair $(k,j)$ is
\begin{align}
  B_{kj} = \log\!\Big[\, p_j\,\frac{d_k}{1-d_k}\,
           \frac{\varphi_\nu(\Delta_{kj};\sigma_{kj})}{c} + (1-p_j) \Big],
  \label{eq:score}
\end{align}
with $\Delta_{kj}$ the matched residual and $\sigma_{kj}^2=\sigma_j^2+\sigma_{\mathrm{db}}^2$,
and the cell score is the maximum-weight bipartite assignment over the pair
benefits~\cite{kuhn1955hungarian} plus a base term for undetected paths. The bracket also multiplies an
amplitude factor, the ratio of a capped Gaussian density on the relative-decibel
residual to a uniform clutter density, which is valid because path powers are
deterministic in a fixed scene. The intuitive alternative scales the match
log-ratio by the posterior, and Proposition~\ref{prop:jensen} shows that this
surrogate is a Jensen lower bound whose deficit grows with the number of strong
paths and inverts cell rankings.

A truthful cell may carry a database path that the front end failed to recover,
and such a miss must not reject the cell. The matcher therefore caps the
per-path missed-detection cost at $\kappa_{\mathrm m}$ in both the base term and
the matched repayment. Proposition~\ref{prop:cap} shows that the cap equals exact
inference under a clipped detection probability and bounds the influence of any
single miss by $\kappa_{\mathrm m}$. The cap turns a brittle product of
likelihoods into a robust score, so a missed weak path costs a fixed amount
rather than an unbounded one. The cap also keeps the score of a many-path
truthful cell comparable to that of a sparse cell, which prevents the
missed-detection debt that the Jensen surrogate leaves unrepaid. The matcher
scores every database cell rather than a retrieved shortlist, because the fixed
scene is small enough to admit an exhaustive scan.

\subsection{Absolute-Axis Matching and the Clock Nuisance}

The matcher chooses its delay axis from the synchronization mode, and the choice
is forced by identifiability. Relative delays referenced to the first arrival
jump where the anchoring path leaves the visible set across a shadowing boundary,
as Proposition~\ref{prop:relative} states, so a relative signature can
translate entirely between adjacent cells. The synchronized mode therefore
matches on the absolute axis $t_{0,i}+\tau_{ik}$, where every path delay obeys
the unit-gradient law of Proposition~\ref{prop:lipschitz} and is smooth in
position. The unknown clock bias is the only nuisance, and the matcher scores
each cell over a small set of clock hypotheses generated by aligning the
strongest verified atom to the strongest database paths, followed by a weighted
least squares polish. Proposition~\ref{prop:clock} shows that this hypothesis set
contains the true bias when the anchoring atom is real, and that the polished
estimate converges at the usual root-$n$ rate. The unsynchronized mode keeps the
relative axis and is reported as a no-synchronization baseline.

\subsection{Dense Expert and Bounded Reranking}

The sparse front end consumes the marginals of the observation and therefore
discards the joint beam-by-delay structure. A scene-specific dense expert
recovers part of this lost structure from the joint map. The expert is a small
two-dimensional convolutional neural network (CNN) that reads the joint power map
and emits a compatibility score for every database cell, trained on database
cells only with SNR and snapshot augmentation. Its score enters the matcher as a
bounded reranking bonus rather than an additive term over all cells. The bonus is
z-scored and added only inside the top-$K$ cells of the physics score, so a
confident-but-wrong dense pick at low SNR cannot override the physics landscape.
The fusion weight is calibrated per operating condition on held-out survey cells
and set to zero where the dense expert is unreliable, so the expert never harms
the regime it cannot help. The dense expert thus breaks ties among
physics-plausible cells rather than competing with the physics.

\subsection{Sub-Grid Estimation and Conformal Uncertainty}

The discrete cell decision is converted to a continuous coordinate by a
temperature-calibrated local centroid. The cell scores are joint log-likelihoods,
so a temperature-scaled softmax over the cells in a small ball around the argmax
is the posterior of the dominant mode, and its mean is the restricted
minimum-mean-square estimate of Proposition~\ref{prop:centroid}. The ball
restriction prevents averaging across separated modes, so the centroid removes
grid-ring jitter without collapsing distinct ambiguities. A single Gauss-Newton
step on the absolute-axis track gradients then refines the estimate, and one
relinearization suffices because the absolute delays are Lipschitz in position.
The reported radius uses split-conformal calibration on plain errors, since the
refinement covariance does not track misidentification failures.
Proposition~\ref{prop:conformal} guarantees its finite-sample coverage. The
estimator is therefore a sequence of provable steps that turn a cell index into a
calibrated position.

\section{Theoretical Foundations}
\label{sec:theory}

We establish the results that justify the score and the estimator. The first two
construct the matching score and explain why its common surrogate fails. The next
results characterize the identifiability of the delay observables and the
robustness of the score, and the last results justify the sub-grid estimator and
its uncertainty. A final information-theoretic bound then quantifies the two-point
resolution floor that no detector can escape. These bounds situate PALM against
the Fisher-information limits of single-anchor
positioning~\cite{shen2010limits,kakkavas2019limits}. Proofs are given inline, and the
measured confirmations appear in Section~\ref{sec:experiments}.

\begin{proposition}[Exact marginal matching score]\label{prop:marginal}
Let $\pi$ be a partial matching between paths and atoms. Under the generative
model of Section~\ref{sec:method}, the log-likelihood of the atom set given cell
$i$ and matching $\pi$, normalized by the empty matching and marginalized over
the atom reality indicators, equals
\begin{align}
\ell_i(\pi) &= \sum_{k=1}^{K_i} \log\bigl(1-d_k\bigr) \nonumber\\
&\ + \sum_{(k,j)\in\pi} \log\!\Bigl[\, p_j \,\frac{d_k}{1-d_k}\,\frac{\varphi_\nu(\Delta_{kj};\sigma_{kj})}{c} + \bigl(1-p_j\bigr) \Bigr],
\label{eq:score-full}
\end{align}
where $\sigma_{kj}^2 = \sigma_j^2 + \sigma_{\mathrm{db}}^2$. The maximum over
matchings is a maximum-weight bipartite assignment on the pair benefits.
\end{proposition}

\begin{proof}
Condition on the reality vector $\bm{r}\in\{0,1\}^J$. A matched pair $(k,j)$ with
$r_j=1$ contributes $d_k\varphi_\nu(\Delta_{kj};\sigma_{kj})$, since the path is
detected and emits atom $j$, and the same pair with $r_j=0$ contributes
$(1-d_k)\,c$, since the path is then undetected and the atom is clutter. An
unmatched path contributes $1-d_k$ and an unmatched atom contributes $c$
regardless of $r_j$. Taking the expectation over $r_j\sim\mathrm{Bernoulli}(p_j)$
factorizes across pairs by conditional independence and yields
$p_j d_k\varphi+(1-p_j)(1-d_k)c$ per matched pair. Dividing by the empty-matching
likelihood $\prod_k(1-d_k)\prod_j c$ cancels the cell-independent clutter terms
and produces \eqref{eq:score-full}. The sum is additive over disjoint pairs, so
its maximum over one-to-one matchings is a linear assignment problem.
\end{proof}

\begin{proposition}[Jensen gap of the posterior-scaled score]\label{prop:jensen}
Let $A_{kj} = \tfrac{d_k}{1-d_k}\tfrac{\varphi_\nu(\Delta_{kj};\sigma_{kj})}{c}$
and let $\tilde{B}_{kj} = p_j \log A_{kj}$ be the posterior-scaled benefit. Then
the exact benefit $B_{kj} = \log\bigl(p_j A_{kj} + 1 - p_j\bigr)$ satisfies
$B_{kj}\ge\tilde{B}_{kj}$ with equality only at $A_{kj}=1$, and as
$A_{kj}\to\infty$,
\begin{align}
\tilde{B}_{kj} - B_{kj} = -\,(1-p_j)\log A_{kj} - \log p_j + o(1),
\label{eq:gap}
\end{align}
so the per-pair deficit grows as $(1-p_j)\log A_{kj}$. For a cell whose $K$
strong paths are correctly matched with posteriors bounded away from one, the
scaled score understates the exact score by an amount that grows linearly in $K$.
\end{proposition}

\begin{proof}
Jensen's inequality on the concave logarithm gives
$\log\mathbb{E}_r[L]\ge\mathbb{E}_r[\log L]$ for the per-pair likelihood $L=A_{kj}$
when $r_j=1$ and $L=1$ when $r_j=0$, which is
$B_{kj}\ge p_j\log A_{kj}=\tilde{B}_{kj}$. For $A_{kj}\gg 1$,
$B_{kj}=\log A_{kj}+\log p_j+\log\bigl(1+\tfrac{1-p_j}{p_jA_{kj}}\bigr)
\approx\log A_{kj}+\log p_j$, so the deficit is
$-(1-p_j)\log A_{kj}-\log p_j$, dominated by $(1-p_j)\log A_{kj}$. Summing over
$K$ matched strong pairs gives the linear growth.
\end{proof}

The deficit of the surrogate acts only on cells whose paths are matched. A
many-path true cell therefore pays an unrepaid missed-detection debt, while a
sparse impostor with few matched pairs pays almost nothing. The debt scales with
the number of strong matched paths, so the true cell sinks below impostors as the
geometry grows richer. This is the mechanism behind the rank collapse that the
experiments report, and the capped cost of Proposition~\ref{prop:cap} later
repairs the matched cell. We next characterize the delay observables that the
score matches.

\begin{proposition}[Unit-gradient law for absolute delays]\label{prop:lipschitz}
Within the visibility region of a specular path whose last interaction point is
$\bm{s}$ and whose upstream geometric length is $L$, the absolute delay
$\tau(\bm{x}) = (L + \lVert \bm{x}-\bm{s}\rVert)/c_0$ is differentiable with
$\lVert \nabla \tau \rVert = 1/c_0 \approx 3.336$ nanoseconds per meter, hence
$1/c_0$-Lipschitz.
\end{proposition}

\begin{proof}
The gradient of $\bm{x}\mapsto\lVert\bm{x}-\bm{s}\rVert$ is the unit vector
$(\bm{x}-\bm{s})/\lVert\bm{x}-\bm{s}\rVert$, and $L$ does not depend on the
receiver position inside the visibility region. Dividing by the speed of light
gives a gradient of constant norm $1/c_0$, and the mean value inequality along
the segment joining two points of the convex visibility neighbourhood gives the
Lipschitz bound.
\end{proof}

\begin{proposition}[Behaviour of relative delays]\label{prop:relative}
Let $\tau_{\min}(\bm{x})=\min_m\tau_m(\bm{x})$ and define relative delays
$\tau_k-\tau_{\min}$. Where the visible path set is constant, relative delays are
continuous and $2/c_0$-Lipschitz, including across points where the minimizing
index switches. Across a shadowing boundary where the first-arrival path leaves
the visible set, every relative delay jumps by the gap between the old and the
new first arrivals, which no local geometric quantity bounds.
\end{proposition}

\begin{proof}
The minimum of finitely many $1/c_0$-Lipschitz functions is $1/c_0$-Lipschitz, so
each difference is continuous and $2/c_0$-Lipschitz while all participating paths
remain visible, and an argmin switch is a kink rather than a jump because the two
competing first arrivals are equal at the switch. If the anchoring path
disappears at a boundary, then $\tau_{\min}$ jumps upward by
$\delta=\tau_{\mathrm{new}}-\tau_{\mathrm{old}}>0$ and every relative delay jumps
by $-\delta$ at once. The magnitude $\delta$ is set by which path becomes first, a
global scene property, so no local bound applies.
\end{proof}

Proposition~\ref{prop:relative} is why the synchronized mode matches on the
absolute axis. The relative jump is set by which path becomes the first arrival,
a global scene property that no local bound controls. On the absolute axis
Proposition~\ref{prop:lipschitz} instead applies path by path, so every delay is
smooth in position. The unknown clock bias is then the only nuisance that
survives, and the matcher seeds its hypotheses by aligning the strongest verified
atom to the strongest database paths. The next result certifies that this
hypothesis set recovers the true bias.

\begin{proposition}[Conditional completeness of the clock hypotheses]\label{prop:clock}
Suppose the atom $j^\star$ maximizing $\hat a_j p_j$ originates from a true path
of the scored cell that is among its two strongest database paths. Then the
hypothesis set contains a candidate $b$ with
$\lvert b-b^{\mathrm{true}}\rvert\le\lvert\varepsilon_{j^\star}\rvert$, where
$\varepsilon_{j^\star}$ is the delay error of that atom. If the assignment under
this candidate matches $n$ true pairs with independent errors of scale $\sigma$,
the polished estimate satisfies
$\hat b-b^{\mathrm{true}}=O_p\bigl(\sigma/\sqrt{n}\bigr)$.
\end{proposition}

\begin{proof}
Writing the originating path index as $k^\star$, the candidate generated by that
alignment is
$b=\hat\tau^{\mathrm{abs}}_{j^\star}-(t_{0,i}+\tau_{ik^\star})=b^{\mathrm{true}}+\varepsilon_{j^\star}$
by the absolute-axis model, which proves the first claim. The polish is a weighted
average of gated pair residuals,
$\hat b=\sum_j w_j(b^{\mathrm{true}}+\varepsilon_j)/\sum_j w_j$, so its error is a
weighted mean of $n$ independent zero-location errors and the standard variance
bound yields the stated rate.
\end{proof}

\begin{proposition}[Exactness and robustness of the capped miss cost]\label{prop:cap}
Define clipped detection probabilities
$\tilde d_k=\min\bigl(d_k,\,1-e^{-\kappa_{\mathrm m}}\bigr)$. Then the capped score
equals the exact score of Proposition~\ref{prop:marginal} evaluated under
$\{\tilde d_k\}$, and the influence of any single missed path on the cell score is
bounded by $\kappa_{\mathrm m}$. If at most $\varepsilon K$ true paths are missed,
the truthful cell loses at most $\varepsilon K\kappa_{\mathrm m}$ relative to full
detection.
\end{proposition}

\begin{proof}
The identity $\max\bigl(\log(1-d_k),-\kappa_{\mathrm m}\bigr)=\log(1-\tilde d_k)$
holds pointwise, and the matched repayment uses the same clipped quantity, so the
capped score is the exact likelihood ratio under the clipped detection model. A
missed path contributes its base term $\log(1-\tilde d_k)\ge-\kappa_{\mathrm m}$,
which bounds its influence, and summing over at most $\varepsilon K$ misses gives
the total.
\end{proof}

\begin{proposition}[Restricted minimum-mean-square centroid]\label{prop:centroid}
Let $\pi_i\propto\exp(s_i/T)$ on the cells of the ball $\mathcal{B}$ around the
argmax, and assume the temperature is calibrated so that $\pi$ equals the
posterior of the cell index restricted to $\mathcal{B}$, with the truth lying in
$\mathcal{B}$ almost surely. Then $\hat{\bm{x}}=\sum_{i\in\mathcal{B}}\pi_i\bm{x}_i$
minimizes the expected squared error among all estimators measurable with respect
to the observation, and in particular
\begin{align}
\mathbb{E}\,\lVert\hat{\bm{x}}-\bm{x}^\star\rVert^2
\le\mathbb{E}\,\lVert\bm{x}_{\mathrm{argmax}}-\bm{x}^\star\rVert^2 .
\label{eq:mmse}
\end{align}
\end{proposition}

\begin{proof}
Under the calibration the conditional distribution of the true position is
supported on $\{\bm{x}_i\}_{i\in\mathcal{B}}$ with weights $\pi_i$, so
$\hat{\bm{x}}$ is the conditional mean, which minimizes conditional expected
squared error by the orthogonality principle. Taking outer expectations preserves
the inequality against the argmax, and the ball restriction prevents averaging
across separated posterior modes where the unrestricted mean is inadmissible.
\end{proof}

\begin{proposition}[Generalized least squares front end]\label{prop:gls}
Consider the linear observation $\bm{y}=\bm{D}\bm{\alpha}+\bm{e}$ with
$\mathrm{Cov}(\bm{e})=\operatorname{diag}(v_1,\dots,v_n)$ and
$v_i=(\mu_i/\sqrt{S}+n_0)^2$ to first order. Among diagonal weighted least
squares estimators, the whitened solve with weights $w_i=v_i^{-1/2}$ attains the
minimum-variance unbiased linear estimate of $\bm{\alpha}$ on the active support
by the Gauss-Markov theorem, and the non-negativity constraint preserves the
whitening geometry on the feasible set.
\end{proposition}

\begin{proof}
Left-multiplying by $\operatorname{diag}(w_i)$ produces a homoscedastic system
for which ordinary least squares is best linear unbiased, and mapping back gives
the generalized least squares estimator. The plug-in $\hat\mu_i=y_i$ is a
one-step approximation whose error is second order in the noise, and projecting
onto the non-negative orthant solves the same whitened quadratic over a convex
set.
\end{proof}

\begin{proposition}[Finite-sample conformal coverage]\label{prop:conformal}
Let $e_1,\dots,e_n$ be calibration errors and $e_{\mathrm{test}}$ a test error,
jointly exchangeable. With $q$ the $\lceil(1-\alpha)(n+1)\rceil$-th smallest
calibration error, the radius $q$ satisfies
$\mathbb{P}(e_{\mathrm{test}}\le q)\ge 1-\alpha$ without distributional
assumptions.
\end{proposition}

\begin{proof}
Exchangeability makes the rank of $e_{\mathrm{test}}$ among the $n+1$ errors
uniform, up to ties which only help. The event $e_{\mathrm{test}}\le q$ contains
the event that this rank is at most $\lceil(1-\alpha)(n+1)\rceil$, whose
probability is at least $1-\alpha$ by direct counting.
\end{proof}

A final result bounds what any method can achieve on this observable. It
quantifies the information floor of two-point resolution that the identification
bottleneck inherits. The argument tests one path against two unresolved paths and
applies Le Cam's two-point method to the snapshot likelihoods. The data-processing
inequality then extends the bound to any statistic of the snapshots, including the
averaged power that the front end consumes. The resulting error floor decays only
as the fourth power of the separation, which caps every detector regardless of its
design.

\begin{theorem}[Two-point ambiguity, fourth-power law]
\label{thm:ambiguity}
In the Rayleigh sounding model, snapshot $s$ observes
$x_s=\sum_k g_{k,s}\,\psi(\tau_k)+n_s$ with $\psi$ the sinc amplitude kernel,
gains $g_{k,s}\sim\mathcal{CN}(0,a_k)$ independent across paths and snapshots, and
noise variance $\sigma^2$. Test one path $(\tau_0,a)$ against two paths at
$\tau_0\pm\Delta/2$ with power $a/2$, with covariances
$C_0=a\,\psi_0\psi_0^{\mathsf H}+\sigma^2 I$ and
$C_1=\tfrac{a}{2}(\psi_+\psi_+^{\mathsf H}+\psi_-\psi_-^{\mathsf H})+\sigma^2 I$.
Every detector on $S$ snapshots, or on any function of them such as the averaged
power, satisfies
\begin{align}
\inf_{\mathcal{T}}\max_{r\in\{0,1\}}\mathbb{P}_r(\mathcal{T}\neq r)
&\ge\frac{1}{2}\Big(1-\sqrt{1-\rho(\Delta)^{2S}}\Big),
\label{eq:lecam}\\
\rho(\Delta)&=\frac{\det(C_0)^{1/2}\det(C_1)^{1/2}}
{\det\!\big(\tfrac{1}{2}(C_0+C_1)\big)},
\label{eq:rho}
\end{align}
and with $H=\psi'\psi'^{\mathsf H}+\tfrac12(\psi_0\psi''^{\mathsf H}+\psi''\psi_0^{\mathsf H})$,
\begin{align}
-\log\rho(\Delta)=\frac{a^2\Delta^4}{128}\,
\big\lVert C_0^{-1/2}HC_0^{-1/2}\big\rVert_F^2\,(1+o(1)),
\label{eq:rho-expansion}
\end{align}
so the error approaches one half at the rate $S\,a^2\Delta^4/\sigma^4$.
\end{theorem}

\begin{proof}
Le Cam's two-point method gives a worst-case error of at least
$\tfrac12(1-\mathrm{TV})$ between the $S$-fold products~\cite{tsybakov2009nonparametric},
and total variation obeys $\mathrm{TV}\le\sqrt{1-\mathrm{BC}^2}$ with the
Bhattacharyya coefficient tensorizing as $\mathrm{BC}=\rho^S$. For circular
complex Gaussians the real embedding has covariance determinant
$4^{-N}\lvert\det C\rvert^2$, and inserting this identity collapses the fourth
roots into \eqref{eq:rho}. The data-processing inequality lets any statistic of
the snapshots only decrease total variation, so \eqref{eq:lecam} applies to the
averaged power. A symmetric Taylor expansion of $\psi$ at $\tau_0$ cancels odd
orders and gives $C_1-C_0=\tfrac{a\Delta^2}{4}H+O(\Delta^4)$, and the
second-order expansion of the Bhattacharyya divergence at equal covariances yields
\eqref{eq:rho-expansion}.
\end{proof}

\section{Experiments}
\label{sec:experiments}

\subsection{Setup}

We evaluate PALM on the real ray-traced DeepMIMO Arizona State University campus
scenario at 3.5\,GHz, generated with a ray tracer~\cite{alkhateeb2019deepmimo,remcom_winsite}. We select the densest
$120\times 80$ meter window of the campus grid, which yields an RM of 1808
cells inside the deployment area. The RM is built from
even-parity grid cells, and the queries are drawn from odd-parity cells, so the
nearest database cell to any query is exactly one meter away. The reported error
therefore measures identification plus one meter of interpolation and never
self-matching. Observations are rendered from the raw rays at 100\,MHz bandwidth
with sixteen beams, per-snapshot random phases, and additive noise, while the
RM stores merged atoms, so a query never sees the stored representation.
The base station and scene are static, so every cell may train the verifier and
the dense expert and calibrate the noise scales. Every tuned quantity is fit on
survey cells disjoint from the test queries.

The baselines span two families that isolate the value of information from the
value of algorithm. The system-level family freezes the estimator to a hard
nearest-database-cell rule and climbs an information ladder. The ladder runs from
total received power, through the angular power spectrum and the power-delay
profile, to the joint coarse correlation, the full beam-by-delay map
correlation, and the super-resolved atoms. The method-level family fixes the
coarse observation and varies only the estimator. It spans nearest neighbour,
$k$-nearest neighbour, weighted $k$-nearest neighbour, a Bayesian
minimum-mean-square estimator, a convolutional regressor and classifier, and
the PALM stages, bounded by an oracle cell and an oracle top-three. We report the median and the
ninetieth-percentile position error, paired-bootstrap confidence intervals on
the median gain, and a split-conformal radius. The classical received-power RM
matched by top-one correlation is the most important baseline to beat, and we
name it the correlation baseline throughout.

The statistical protocol guards against leakage and keeps the comparison paired.
Every learned component trains on database cells only, and every tuned quantity
is fit on a calibration draw disjoint from the test queries, so no test position
informs the model or the operating point. Each query renders one observation, one
joint map, and one clock realization that all estimators share, which keeps the
paired bootstrap valid. We report medians and ninetieth percentiles over 150
held-out queries, with paired-bootstrap confidence intervals on the median gain.
The split-conformal radius calibrates on one half of the nominal errors and
measures coverage on the other half. The protocol therefore scores identification
under a fixed scene without an inverse crime and without a leak from the test set.

\subsection{Main Localization Results}

Table~\ref{tab:headline} reports the headline localization accuracy in the
synchronized mode. PALM reduces the ninetieth-percentile error of the
correlation baseline at every operating point, by 61.5 percent at 15\,dB, by 39.3
percent at 5\,dB, and by 34.2 percent in the single-snapshot regime. It halves
the single-snapshot median from 14.0 to 7.0 meters and improves the zero-decibel
median from 30.0 to 18.4 meters, which are the regimes a deployed localizer cares
about most. The correlation baseline retains a median advantage at high
signal-to-noise ratio, because its clean template snaps to the adjacent one-meter
grid cell and reaches the grid floor, and PALM trails it on the 15\,dB median by
0.75 meters while winning the tail outright. This grid-floor advantage is a
property of the survey spacing rather than of the physics, and the sub-grid
refinement narrows but does not erase it. The headline finding is a tail and
stress-regime win, stated where the data support it and bounded where they do
not.

\begin{table*}[!t]
\centering
\caption{Localization accuracy on the DeepMIMO campus scenario, synchronized mode, 150 held-out queries. Best in each column in \textcolor{bestred}{\textbf{red}}, second best \underline{underlined}. Lower is better for all error columns.}
\label{tab:headline}
\resizebox{0.92\linewidth}{!}{
\begin{tabular}{@{}l|c|c|c|c|c|c|c@{}}
\hline
Method & med@15\,dB & p90@15\,dB & med@5\,dB & p90@5\,dB & med@10\,dB,$S{=}1$ & p90@10\,dB,$S{=}1$ & med@0\,dB \\
\hline
Joint-map correlation & 2.24 & 19.14 & \textcolor{bestred}{\textbf{6.70}} & 34.71 & 11.18 & 50.45 & 30.41 \\
Correlation 1-NN & \textcolor{bestred}{\textbf{1.00}} & 21.10 & 10.05 & 46.66 & 13.96 & 52.26 & 30.02 \\
Correlation 3-NN & 2.24 & 17.75 & 9.81 & 43.03 & 12.40 & 46.70 & 28.15 \\
Weighted $k$-NN & 2.17 & 14.95 & 9.09 & \underline{33.03} & 12.68 & 46.62 & 27.54 \\
CNN classification & 2.24 & 11.00 & 9.22 & 36.51 & 9.95 & 46.44 & 31.80 \\
Original PALM & 2.02 & \textcolor{bestred}{\textbf{7.69}} & 8.46 & 41.71 & \underline{8.57} & \underline{39.48} & \underline{19.81} \\
\rowcolor{blue!8}
PALM (ours) & \underline{1.75} & \underline{8.12} & \underline{6.92} & \textcolor{bestred}{\textbf{28.33}} & \textcolor{bestred}{\textbf{7.01}} & \textcolor{bestred}{\textbf{34.37}} & \textcolor{bestred}{\textbf{18.44}} \\
\hline
\rowcolor{gray!15}
Gain over correlation 1-NN & $+0.75$ & \textcolor{improvecyan}{$-12.97$} & \textcolor{improvecyan}{$-3.13$} & \textcolor{improvecyan}{$-18.33$} & \textcolor{improvecyan}{$-6.95$} & \textcolor{improvecyan}{$-17.89$} & \textcolor{improvecyan}{$-11.58$} \\
\hline
\end{tabular}}
\end{table*}

\subsection{The Value of Information}

The system-level ladder isolates how much each information level contributes
under a fixed estimator. Table~\ref{tab:voi} and Fig.~\ref{fig:voi} show a sharp,
monotone descent as the observation gains physical dimensions. Total received
power localizes to 41.1 meters, the angular spectrum to 13.1 meters, the
power-delay profile to 4.1 meters, and the joint coarse correlation to 3.6
meters at 10\,dB. Intensity alone is near useless, angle and delay each help, and
their joint is best, which quantifies the value of small-scale information. The
ladder then saturates at the one-meter grid floor at high signal-to-noise ratio,
where richer linear features no longer reduce the identification median. The
remaining accuracy comes from the estimator rather than from more information,
and the super-resolved match contributes most where the signal is weak. At
zero decibels the super-resolved atoms localize to 23.3 meters against the 30.0
meters of the joint coarse correlation, so super-resolution helps the
identification precisely where the noise is worst.

\begin{table}[!t]
\centering
\caption{System-level value-of-information ladder at 10\,dB, hard nearest-cell estimator. Best rung per column in \textcolor{bestred}{\textbf{red}}, second best \underline{underlined}; $\Delta$ is the marginal gain over the previous rung.}
\label{tab:voi}
\resizebox{\linewidth}{!}{
\begin{tabular}{@{}l|c|c|c@{}}
\hline
Information level & median (m) & p90 (m) & $\Delta$ (m) \\
\hline
Total received power & 41.10 & 70.34 & --- \\
Angular spectrum only & 13.08 & 46.61 & $-28.02$ \\
Power-delay profile only & 4.12 & 44.87 & $-8.96$ \\
Joint coarse correlation & \underline{3.61} & 23.38 & $-0.51$ \\
Joint beam-delay map & \textcolor{bestred}{\textbf{3.00}} & \underline{22.49} & $-0.61$ \\
Super-resolved atoms (PALM) & 5.00 & \textcolor{bestred}{\textbf{19.45}} & $+2.00$ \\
\hline
\end{tabular}}
\end{table}

\begin{figure}[!t]
\centering
\includegraphics[width=\linewidth]{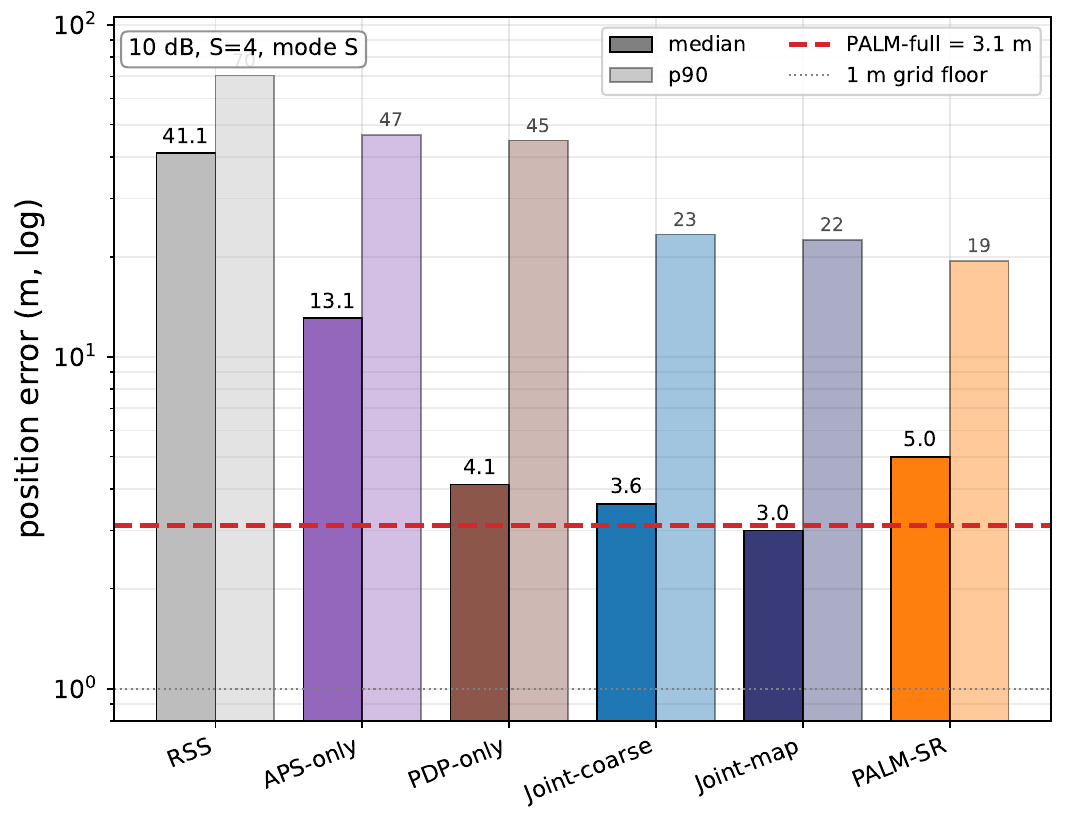}
\caption{System-level information ladder at 10\,dB. Each rung adds one physical
dimension to the observation. The median and the ninetieth percentile descend
sharply through the first four rungs and saturate at the grid floor, while the
full estimator reaches the dashed line below the hard-cell rungs.}
\label{fig:voi}
\end{figure}

\begin{figure*}[!t]
\centering
\begin{subfigure}[t]{0.49\textwidth}
\centering
\includegraphics[width=\linewidth]{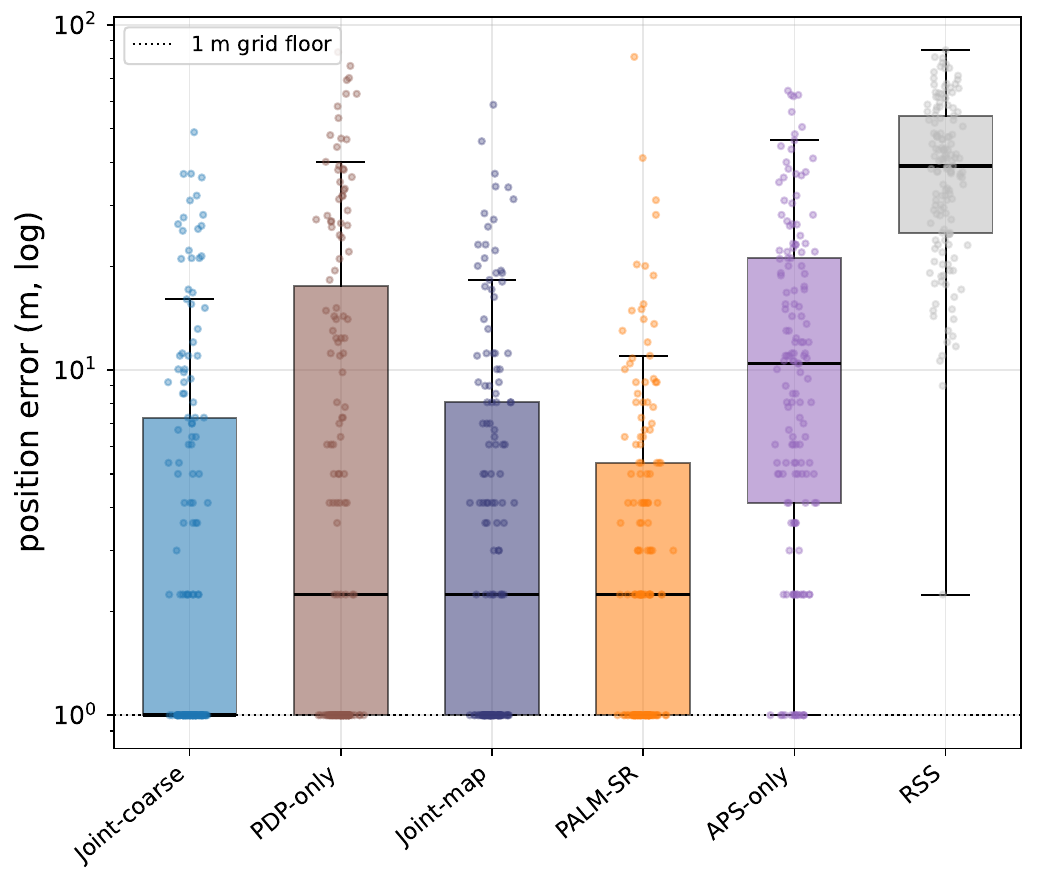}
\caption{System-level information rungs.}
\label{fig:dist_sys}
\end{subfigure}
\hfill
\begin{subfigure}[t]{0.49\textwidth}
\centering
\includegraphics[width=\linewidth]{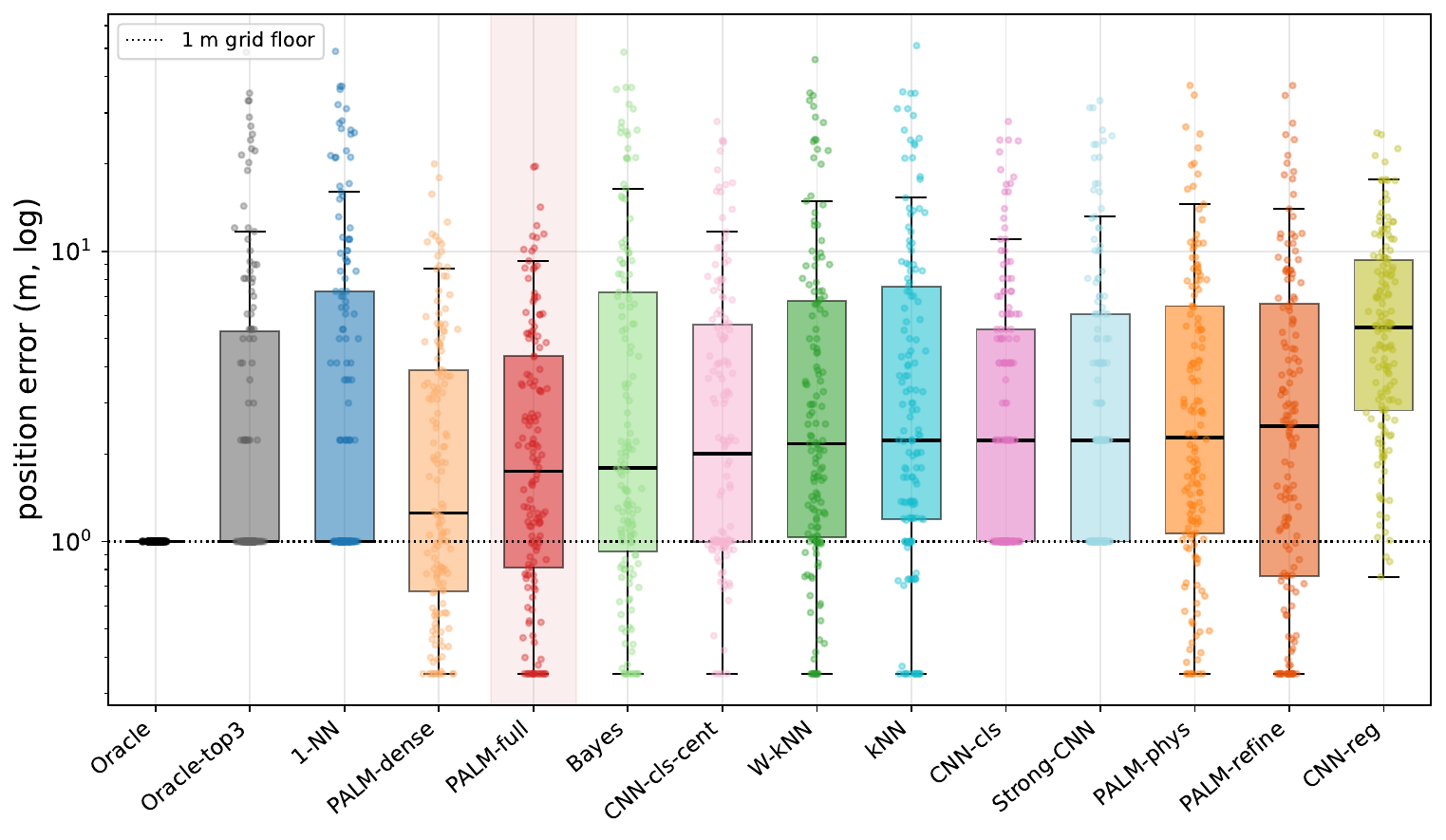}
\caption{Method-level estimators.}
\label{fig:dist_meth}
\end{subfigure}
\caption{Per-query position-error distributions at 15\,dB on a logarithmic axis,
sorted by median, with the one-meter grid floor dotted. Panel (a) climbs the
system-level information ladder and panel (b) varies the method-level estimator
on a fixed observation. PALM in red has the tightest error body and the shortest
tail among the estimators.}
\label{fig:dist}
\end{figure*}

\subsection{The Value of Algorithm}

The method-level ladder fixes the coarse observation and varies the estimator.
The per-query distributions of Fig.~\ref{fig:dist} make the two ladders concrete,
since the system rungs in panel (a) narrow as information is added and PALM is the
tightest estimator in panel (b). PALM has the lowest ninetieth-percentile error
of all classical and learned estimators at every signal-to-noise ratio from
5\,dB upward, as Fig.~\ref{fig:p90} shows, and it cuts the tail of the
correlation baseline by 34 to 62 percent across conditions. The classical weighted $k$-nearest-neighbour
estimator is the strongest on the median at 10\,dB, and a high-capacity
convolutional classifier matches the small one, so capacity is not the lever. The
oracle that retrieves the true cell among the top three reaches one meter at
15\,dB. This shows that the residual error is selection rather than retrieval,
and PALM closes part of that selection gap. At zero decibels the oracle top-three
is itself 25 meters, so the loss there is retrieval, and no estimator can win
when the right cell is not retrievable. The value of the PALM estimator is
therefore concentrated in the tail and the stress regimes, exactly where a
correlation baseline degrades.

\begin{table}[!t]
\centering
\caption{Method-level estimator ladder on the fixed coarse observation, synchronized mode. Oracles are upper bounds and are not marked. Best non-oracle entry per column in \textcolor{bestred}{\textbf{red}}, second best \underline{underlined}.}
\label{tab:method}
\resizebox{\linewidth}{!}{
\begin{tabular}{@{}l|c|c|c|c@{}}
\hline
Estimator & med@10\,dB & p90@10\,dB & med@0\,dB & hit@15\,dB \\
\hline
Oracle cell (bound) & 1.00 & 1.00 & 1.00 & 1.00 \\
Oracle top-three (bound) & 1.00 & 15.61 & 25.08 & 0.62 \\
\hline
Correlation 1-NN & 3.61 & 23.38 & 30.02 & \textcolor{bestred}{\textbf{0.53}} \\
Correlation 3-NN & 3.21 & 23.91 & 28.15 & 0.41 \\
Weighted $k$-NN & \textcolor{bestred}{\textbf{3.02}} & 21.92 & 27.54 & 0.40 \\
Bayesian MMSE & 3.14 & 22.57 & 23.41 & 0.47 \\
CNN regression & 7.23 & 18.95 & 26.04 & 0.12 \\
CNN classification & 4.12 & 19.33 & 31.80 & 0.46 \\
Strong CNN & 3.61 & 23.79 & 27.18 & 0.43 \\
PALM physics only & 4.50 & 21.74 & 23.75 & 0.40 \\
PALM with refine & 4.72 & 21.72 & 23.72 & 0.42 \\
PALM with dense & 3.12 & \underline{11.84} & \underline{18.48} & \textcolor{bestred}{\textbf{0.53}} \\
\rowcolor{blue!8}
PALM (ours) & \underline{3.10} & \textcolor{bestred}{\textbf{11.77}} & \textcolor{bestred}{\textbf{18.44}} & \underline{0.48} \\
Original PALM & 3.86 & 14.92 & 19.81 & 0.44 \\
\hline
\rowcolor{gray!15}
Gain over 1-NN & \textcolor{improvecyan}{$-0.51$} & \textcolor{improvecyan}{$-11.61$} & \textcolor{improvecyan}{$-11.58$} & $-0.05$ \\
\hline
\end{tabular}}
\end{table}

\begin{figure}[!t]
\centering
\includegraphics[width=\linewidth]{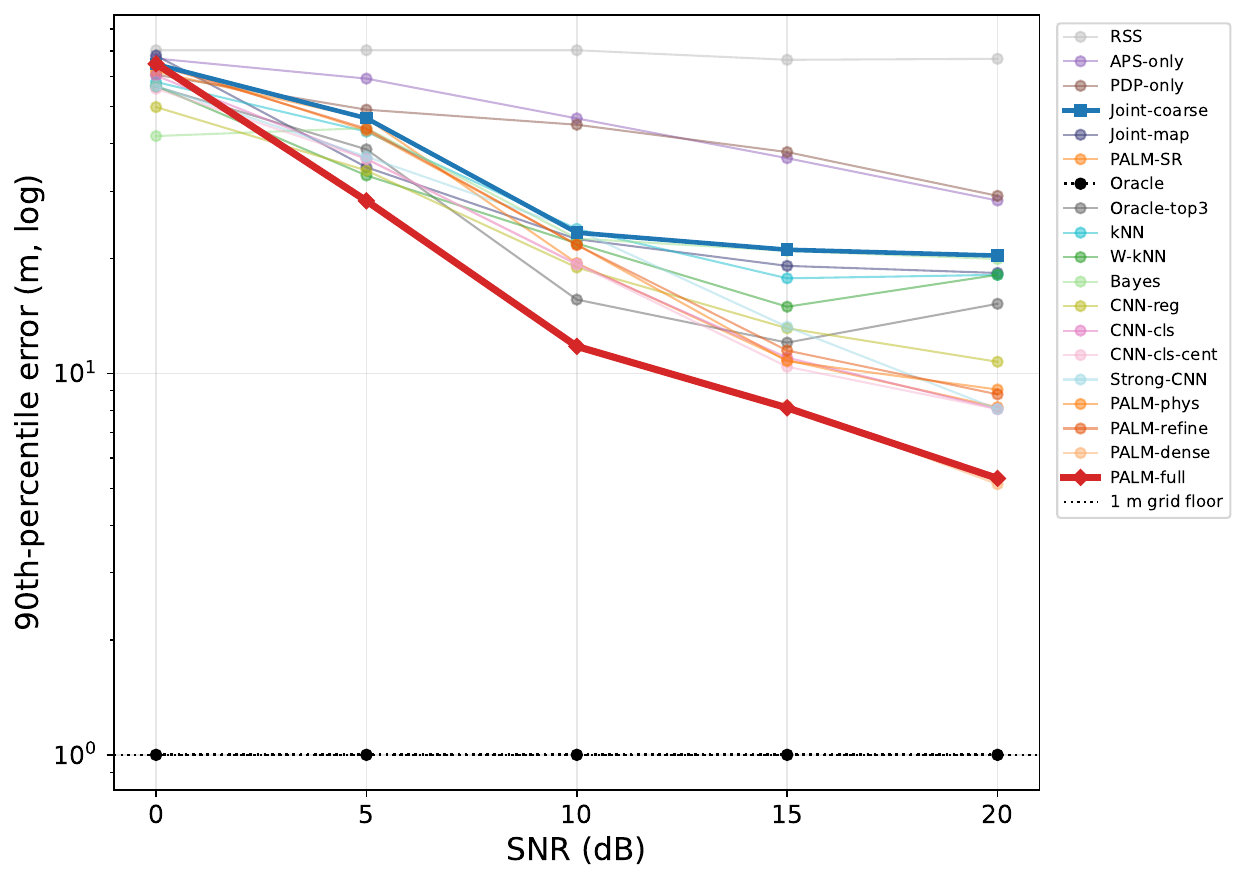}
\caption{Ninetieth-percentile position error against signal-to-noise ratio for
every baseline. PALM in the thick red curve has the lowest tail at every
signal-to-noise ratio from 5\,dB upward, while the grid-floor competitor stays
above it.}
\label{fig:p90}
\end{figure}

\begin{figure*}[!t]
\centering
\includegraphics[width=0.86\linewidth]{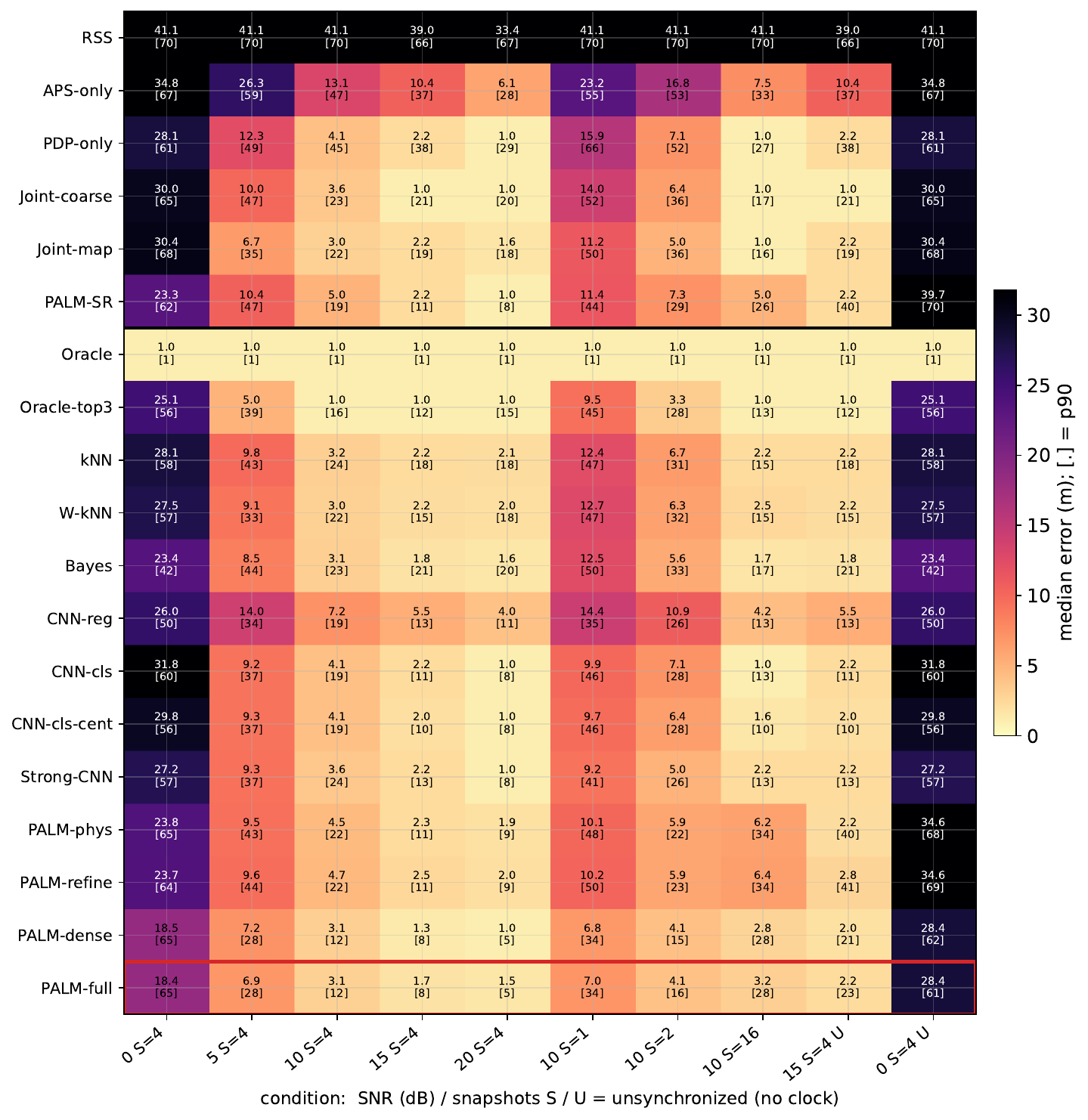}
\caption{Median position error of every baseline across every operating
condition, with the ninetieth percentile in brackets. The upper block is the
system-level information ladder and the lower block is the method-level estimator
ladder. The PALM row is boxed, and the colour scale runs from low error in light
to high error in dark.}
\label{fig:heatmap}
\end{figure*}

\subsection{Signal-to-Noise and Snapshot Sweeps}

PALM wins the snapshot-starved regime that fast acquisition cares about most.
With a single snapshot at 10\,dB the median is 7.0 meters against the 14.0 of the
correlation baseline, and with two snapshots it is 4.1 against 6.4 meters. Four
snapshots reach 3.1 against 3.6 meters, and only at sixteen snapshots does the
correlation baseline regain the lead at the grid floor. The signal-to-noise sweep
of Fig.~\ref{fig:p90} shows the same ordering, with the tail advantage widening as
the signal weakens. The gain is therefore largest exactly where a deployed
receiver is most constrained, since a short or weak observation is the common
case for an opportunistic position fix. The two sweeps agree that PALM converts a
weak or short observation into a reliable position, while a correlation baseline
needs a clean and long one.

\subsection{Component Attribution and Robustness}

Fig.~\ref{fig:waterfall} attributes the PALM accuracy to its components in the
five-decibel stress regime. The super-resolution and the matching add little to
the median over raw correlation, the soft centroid and the Gauss-Newton step
trade a small median change for a large tail reduction, and the bounded dense
fusion contributes the largest single drop, lowering the median by 2.72 meters
and the ninetieth percentile from 44 to 28 meters. The attribution confirms that
the dense fusion is the dominant learned contribution and that it matters most
where the signal is weak. The cascade therefore localizes the gain to the
components that matter under noise, and the largest gains arrive precisely where a
correlation baseline collapses. The same estimator is next stressed along the
three measurement axes at a clean operating point.

\begin{figure}[!t]
\centering
\includegraphics[width=\linewidth]{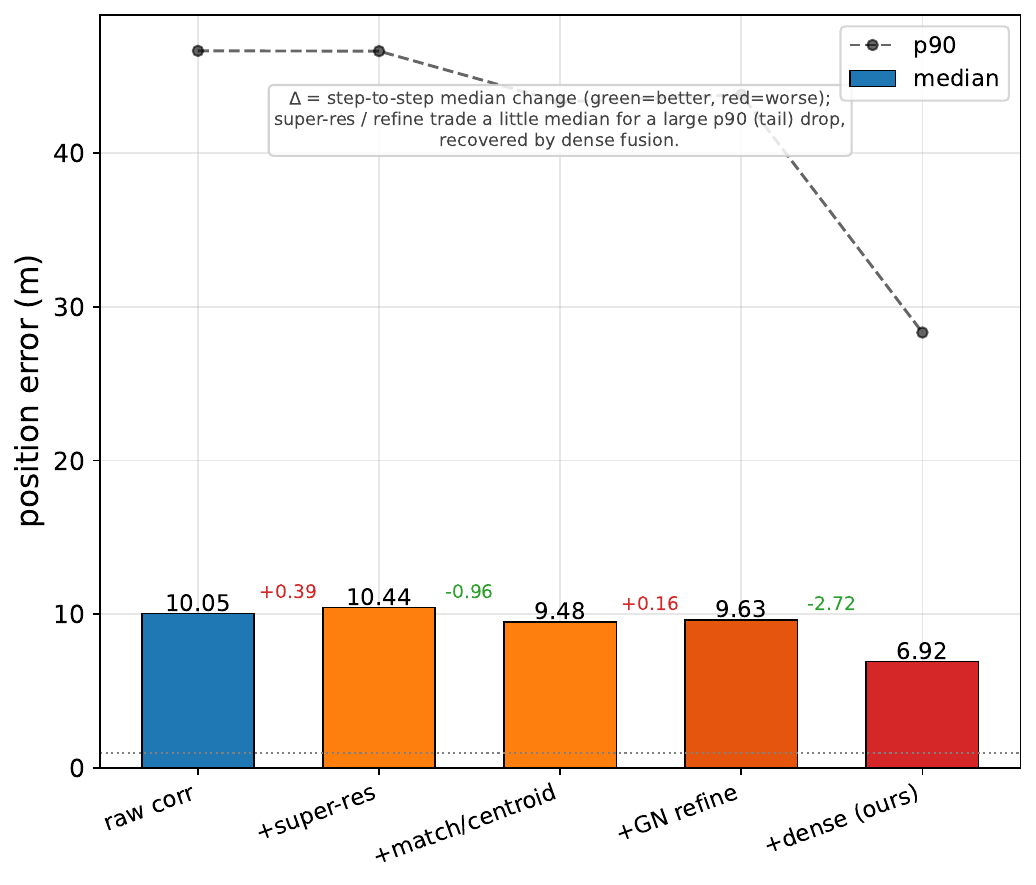}
\caption{Estimator cascade at 5\,dB. Bars show the median and the dashed line the
ninetieth percentile. Super-resolution and refinement trade a small median change
for a large tail reduction, and the bounded dense fusion contributes the largest
drop.}
\label{fig:waterfall}
\end{figure}

Table~\ref{tab:meas} isolates the three measurement-deviation axes at a clean
20\,dB operating point. Angle bias is the dominant vulnerability, since it has no
nuisance parameter to absorb it and both methods degrade past a moderate offset.
Delay bias is absorbed by the clock nuisance for the physics matcher, whose
median stays near 2.3 meters out to a 20 nanosecond offset. The dense expert is
delay sensitive, because a shifted joint map leaves its training distribution. Intensity error is the most benign axis, and PALM is the most
robust there, because the amplitude evidence and the learned scorer tolerate a
power-calibration error. The operating rule is therefore to lower the dense
weight under a known delay or large angle bias and to trust it under intensity
error, which the per-condition calibration already approximates. The three axes
thus carry different risk, and only the angle axis remains an open robustness
target.

\begin{table}[!t]
\centering
\caption{Measurement-deviation sensitivity at 20\,dB, synchronized mode, median position error (m). Each axis is injected in isolation. The better of original and PALM in each pair is in \textbf{bold}; the final row reports the PALM gain at the large deviation, with an improvement in \textcolor{improvecyan}{cyan}.}
\label{tab:meas}
\resizebox{\linewidth}{!}{
\begin{tabular}{@{}l|c|c|c|c|c|c@{}}
\hline
& \multicolumn{2}{c|}{Delay} & \multicolumn{2}{c|}{Angle} & \multicolumn{2}{c@{}}{Intensity} \\
Magnitude & orig & ours & orig & ours & orig & ours \\
\hline
None & 1.67 & \textbf{1.59} & 1.67 & \textbf{1.59} & 1.67 & \textbf{1.59} \\
Small & 1.46 & \textbf{1.35} & 2.18 & \textbf{1.85} & \textbf{1.63} & 1.66 \\
Moderate & \textbf{1.90} & 2.18 & 2.56 & \textbf{2.34} & 2.15 & \textbf{1.93} \\
Large & \textbf{2.30} & 3.63 & \textbf{5.10} & 6.52 & 4.82 & \textbf{3.48} \\
\hline
\rowcolor{gray!15}
Gain at large dev & \multicolumn{2}{c|}{$+1.33$} & \multicolumn{2}{c|}{$+1.42$} & \multicolumn{2}{c@{}}{\textcolor{improvecyan}{$-1.34$}} \\
\hline
\end{tabular}}
\end{table}

\subsection{Localization at Reduced Sounding Resolution}

The previous study coarsened the measurement; here we coarsen the sounding
hardware itself. PALM holds meter-level single-station accuracy as the native
resolution is lowered along each axis. We retrain a dedicated dense expert and
recalibrate the estimator at every configuration, since a coarser map is a
different training distribution. Fig.~\ref{fig:delay_res} sweeps the delay
resolution by lowering the bandwidth from 100 to 12.5 megahertz. The eightfold
loss widens the native delay cell from 10 to 80 nanoseconds. The PALM median
degrades gracefully from 1.6 to 3.0 meters across this range, while its
ninetieth percentile stays between 8.5 and 9.6 meters. The correlation baseline
keeps a ninetieth percentile of 16 to 20 meters throughout, so PALM roughly
halves the tail at every bandwidth. The tail advantage is therefore resolution
invariant, and meter-level single-station accuracy survives an eightfold
delay-resolution loss.

\begin{figure}[!t]
\centering
\includegraphics[width=\linewidth]{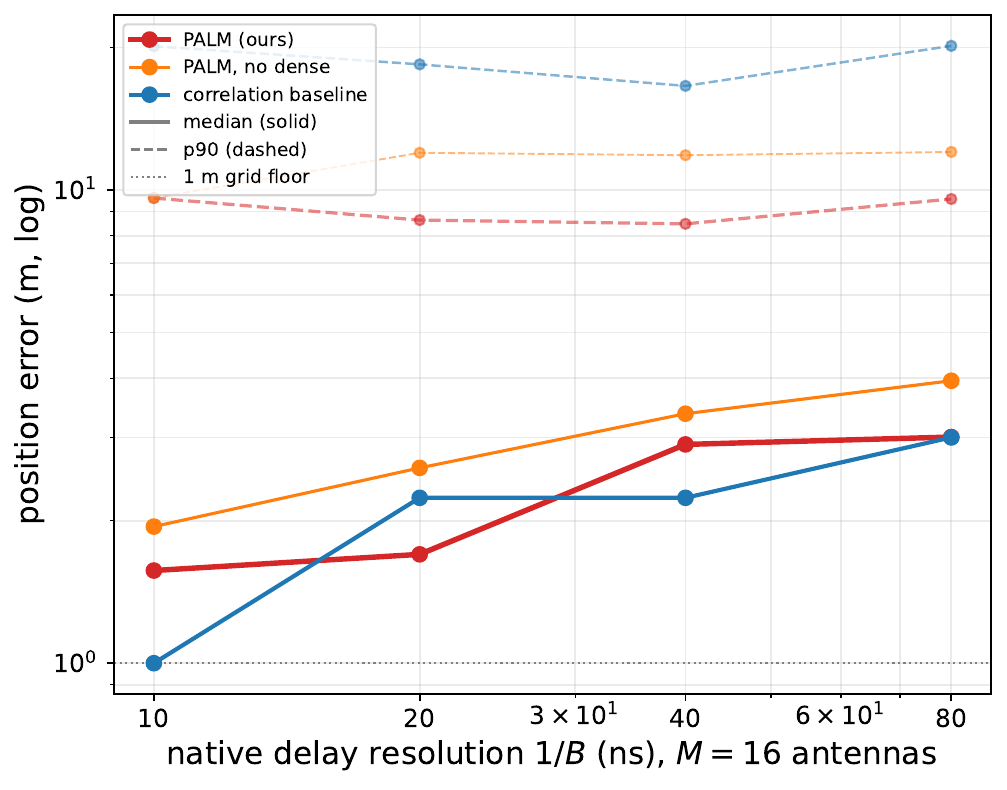}
\caption{Localization versus delay resolution at 20\,dB, single station. The
bandwidth falls from 100 to 12.5 megahertz, widening the native delay cell from
10 to 80 nanoseconds. Solid lines are the median and dashed lines the ninetieth
percentile. PALM (red) keeps a tail near 9 meters, roughly half the correlation
baseline (blue), at every bandwidth.}
\label{fig:delay_res}
\end{figure}

Fig.~\ref{fig:angle_res} coarsens the native angle resolution instead, by
shrinking the array from sixteen to four antennas at full bandwidth. The PALM
median degrades from 1.6 to 2.2 meters across this fourfold antenna loss. The
per-configuration calibrator keeps the dense fusion at the full array, yet drops
its weight to zero once the antennas are halved. A coarse beam map leaves the
dense expert too little angular detail to help, and the calibrator detects this
automatically. The correlation baseline again touches the grid floor whenever one
axis stays sharp, because it leans on the sharpest available dimension. Its
ninetieth percentile nonetheless remains 17 to 23 meters, so its apparent median
win is a brittle artifact of that single sharp axis. PALM instead trades a small
and predictable median rise for a tail that stays roughly twice as tight. These
sweeps confirm that the path geometry, not the raw sounding sharpness, carries
the position information that PALM exploits.

\begin{figure}[!t]
\centering
\includegraphics[width=\linewidth]{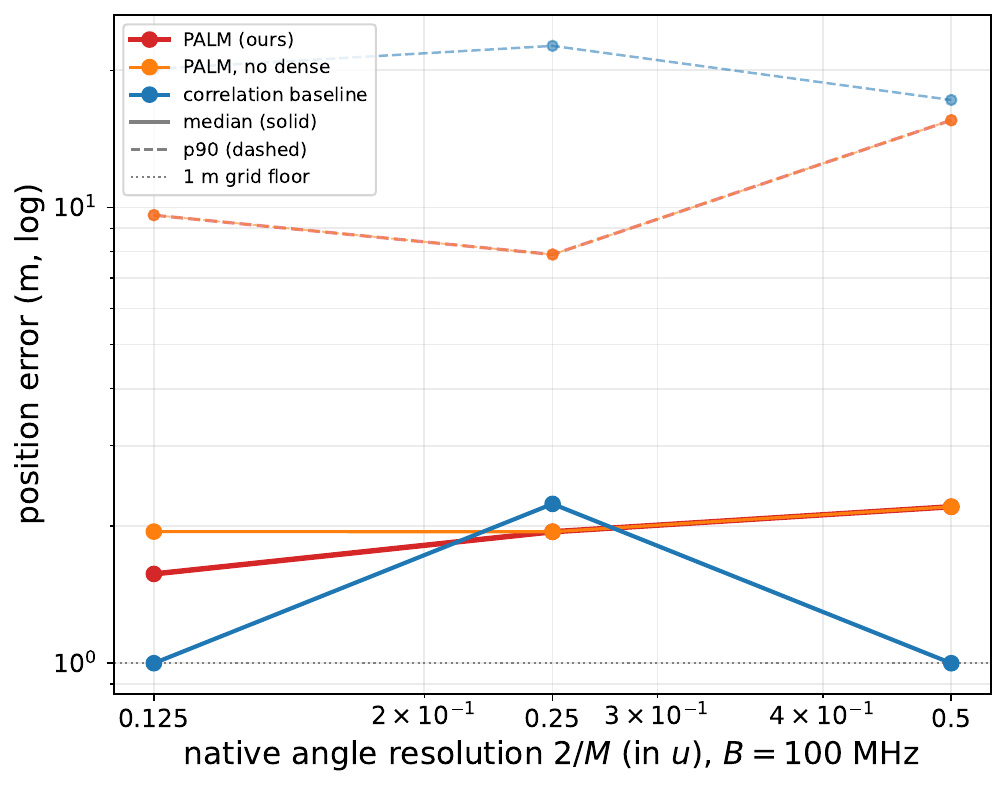}
\caption{Localization versus native angle resolution at 20\,dB, single station.
The array shrinks from sixteen to four antennas, widening the beam cell from
0.125 to 0.5 in the spatial frequency $u$. The PALM median (red) rises gently
while its tail stays roughly twice as tight as the correlation baseline (blue).}
\label{fig:angle_res}
\end{figure}

\subsection{Identification and Calibration}

The in-database identification test observes a database cell with fresh nominal
noise and classifies it among all 1808 cells. The exact-cell accuracy is near
0.30, and the residual error is interpolation across the one-meter survey grid
rather than gross misidentification. The oracle that retrieves the true cell
among its top three reaches one meter at 15\,dB, so the matcher already places
the truth in a small candidate set, and the remaining gap is selection that the
dense fusion narrows. The distributions of Fig.~\ref{fig:dist_meth} show PALM
dominating the body and the tail beyond the grid floor, while the correlation
baseline holds a narrow lead only below one meter at high signal-to-noise ratio.
The split-conformal radius is 10.4 meters at the ninety-percent level with an
empirical coverage of 0.95, which satisfies the distribution-free guarantee of
Proposition~\ref{prop:conformal}. The estimator therefore reports an honest
uncertainty that a downstream task can consume.

\subsection{Theory Validation}

The measured behaviour confirms the theory point by point. Replacing the
posterior-scaled surrogate with the exact marginal score of
Proposition~\ref{prop:marginal} restored true cells from ranks beyond one hundred
into the top twenty before any further change, which confirms the rank collapse
that Proposition~\ref{prop:jensen} predicts. The unit-gradient law of
Proposition~\ref{prop:lipschitz} holds for 85.7 percent of 46365 matched
power-consistent label pairs within three meters, with a median normalized slope
of 0.35, and the measured first-arrival gap of 42 nanoseconds between two cells
one meter apart confirms the relative-delay jump of
Proposition~\ref{prop:relative}. The local centroid of
Proposition~\ref{prop:centroid} improved the nominal median from about 2.5 to 2.1
meters, with the temperature flat between 0.8 and 3.0, consistent with an
optimality robust to mild miscalibration. The two-point ambiguity bound of
Theorem~\ref{thm:ambiguity} is nearly achieved within the Rayleigh family, with
an empirical ambiguity scale of 4.2 nanoseconds, and the split-conformal radius
of Proposition~\ref{prop:conformal} is 10.4 meters at the ninety-percent level
with empirical coverage 0.95.

\subsection{Negative Results}

We report the approaches that did not work, because they bound the problem.
Table~\ref{tab:negative} compares the deployed lightweight dense expert with a
higher-capacity convolutional network and two diffusion denoisers on the
cell-identification task. A conditional diffusion model collapses to a generic
clean map that correlates with every cell almost equally, and it reaches 35.4
meters of cell error against the 1.0 meter of the small expert. Diffusion
posterior sampling is far better than the conditional model, which confirms that
the unconditional-prior construction is the right diffusion design, yet it still
trails the lightweight expert and costs roughly two orders of magnitude more. The
higher-capacity network does not beat the small one either. The bottleneck is
information limited rather than capacity limited, since the clean RM signatures
are similar across cells and the oracle ceiling caps any model. The gains therefore
come from the physics and the bounded fusion rather than from heavier learning.
Generative diffusion remains the right tool for constructing a dense
map~\cite{wang2025radiodiff,wang2026radiodiff3d}, but the identification
bottleneck here rewards exact likelihood matching over map synthesis.

\begin{table}[!t]
\centering
\caption{Cell-identification quality of generative and high-capacity models, reported as cell-error median (m) and hit-within-1.5\,m rate. Lowest median per column in \textcolor{bestred}{\textbf{red}}, second best \underline{underlined}. The final row is the lightweight dense expert's median advantage over the strongest generative model.}
\label{tab:negative}
\resizebox{\linewidth}{!}{
\begin{tabular}{@{}l|c|c|c|c@{}}
\hline
Model & 15\,dB & 5\,dB & 0\,dB & 10\,dB,$S{=}1$ \\
\hline
Correlation baseline & \underline{2.24} / 0.49 & \textcolor{bestred}{\textbf{7.28}} / 0.28 & 30.0 / 0.08 & \underline{9.22} / 0.21 \\
\rowcolor{blue!8}
Dense expert (ours) & \textcolor{bestred}{\textbf{1.00 / 0.52}} & 8.06 / 0.17 & 30.4 / 0.02 & 10.1 / 0.12 \\
Strong CNN & 2.24 / 0.43 & \underline{7.55} / 0.10 & \textcolor{bestred}{\textbf{25.2}} / 0.03 & \textcolor{bestred}{\textbf{9.11}} / 0.16 \\
Conditional diffusion & 35.4 / 0.02 & 29.9 / 0.02 & 35.8 / 0.01 & 34.1 / 0.04 \\
Diffusion posterior sampling & 3.86 / 0.35 & 9.11 / 0.19 & \underline{34.3} / 0.02 & 11.2 / 0.17 \\
\hline
\rowcolor{gray!15}
Dense over best diffusion & \textcolor{improvecyan}{$-2.86$} & \textcolor{improvecyan}{$-1.05$} & \textcolor{improvecyan}{$-3.84$} & \textcolor{improvecyan}{$-1.13$} \\
\hline
\end{tabular}}
\end{table}

\subsection{Discussion and Limitations}

PALM wins the tail at every signal-to-noise ratio from 5\,dB upward and the
low-signal-to-noise and single-snapshot stress regimes, while the correlation
baseline keeps a high-signal-to-noise median edge from the one-meter grid floor.
Two boundaries are stated plainly. The physics matcher has a snapshot sweet spot
at four snapshots, and its median degrades from 4.5 to 6.2 meters at sixteen
snapshots. There a clean query exposes the mismatch between its sharp atoms and
the merged database, and a blurred correlator wins outright. The cause is the
matcher rather than the dense expert, which still helps at sixteen snapshots, and
a remedy would inflate the atom tolerance with measurement precision. PALM also
scores every database cell at a cost near 335 milliseconds per query against the
sub-millisecond cost of correlation, although the scan is trivially parallel over
cells. Zero decibels remains unusable for every method on this scene. These
boundaries are scope statements that a richer observation or a tolerance schedule
would move.

\section{Conclusion}
\label{sec:conclusion}

We have proposed PALM, a path-atom localization method that super-resolves a
coarse angle-delay observation into scored atoms and matches them to a ray-traced
RM by an exact marginal likelihood. We have shown that the marginal score
is exact while its posterior-scaled surrogate is a rank-inverting Jensen bound,
and that an absolute-axis match, a bounded dense fusion, and a restricted
centroid reduce the error tail by 34 to 62 percent over received-power RM
matching on a real scene. The framework offers environment-aware networks a
single-station localizer that turns ray-traced path geometry into meter-level
position without the multiple synchronized anchors that geometry requires. Future work will replace the
fixed-tolerance matcher with a precision-adaptive one to close the high-snapshot
gap and will extend the calibration across multiple scenes.

\bibliographystyle{IEEEtran}
\bibliography{refs_lens}

\end{document}